\documentclass[conference,onecolumn,12pt]{IEEEtran}

\usepackage{array}
\usepackage[utf8]{inputenc}

\usepackage{algorithm,amsbsy,amsmath,amssymb,epsfig,bbm,mathrsfs,multirow,amsthm,algpseudocode}
\usepackage[T1]{fontenc}
\usepackage{subcaption,float}
\usepackage{bm,url}
\usepackage{graphicx}
\graphicspath{{./figures}}

\newtheorem{theorem}{Theorem}

\newtheorem{lemma}{Lemma}

\renewcommand{\baselinestretch}{1.45}

\DeclareMathOperator{\diag}{diag}
\DeclareMathOperator{\rank}{rank}
\DeclareMathOperator{\Tr}{Tr}

\newcommand\numberthis{\addtocounter{equation}{1}\tag{\theequation}}

\begin{document}
%
\title{Managing Interference and Leveraging Secondary Reflections Amongst Multiple IRSs}

\author{\IEEEauthorblockN{Tu V. Nguyen,}
\IEEEauthorblockA{College of Engineering and Computer Science\\
VinUniversity, Ha Noi, Viet Nam\\
Email: tu.nv@vinuni.edu.vn\\}
\and
\IEEEauthorblockN{Diep N. Nguyen,}
\IEEEauthorblockA{School of Electrical and Data Engineering\\
University of Technology Sydney, Australia\\
Email: Diep.Nguyen@uts.edu.au}}


\thispagestyle{plain}
\pagestyle{plain}

\maketitle

\begin{abstract}
Intelligent reflecting surface (IRS) has recently been emerging as an enabler for smart radio environment in which passive antenna arrays can be used to actively tailor/control the radio propagation (e.g., to support users under adverse channel conditions). With multiple IRSs being launched (e.g., coated on various buildings) to support various group of users, it is critical to jointly optimize the phase-shifts of all IRSs to mitigate the interference as well as to leverage the secondary reflections amongst IRSs. This work takes the first step by considering the uplink of multiple users that are grouped and supported by multiple IRSs to a multi-antenna base station. Each IRS with multiple controllable phase-shift elements is intended to serve a group of near-by users. We first formulate the minimum achievable rate (from all users) maximization problem by jointly optimizing phase-shifts of elements from all IRSs and the received beamformers at the MIMO base station. The problem turns out to be non-convex. We then derive its solution using the alternating optimization mechanism. Our simulations show that by properly managing interference and leveraging the secondary reflections amongst IRSs, there is a great benefit of deploying more IRSs to support different groups of users to achieve a higher rate per user. In contrast, without properly managing the secondary reflections, increasing the number of IRSs can adversely impact the network throughput, especially for higher transmit power. 
\end{abstract}


\IEEEpeerreviewmaketitle

\section{Introduction}

Intelligent reflecting surface (IRS) has been emerging as an enabler of smart radio environment in which radio propagation can be deliberately controlled/tailored by passive reflectors. A wireless system assisted by IRS hence can ``reconfigure" the environment to combat the shadowing and fading issues or to create well-scattered environment for spatial multiplexing. An IRS is a planar metasurface that is made up of a large number of passive reflecting elements. Each of these elements can be optimized to alter the amplitude and/or phase shift of the reflected signal onto it. These elements hence together can help effectively ``reshape" the wireless channels \cite{ruizhang2019gaussian}. Such an ability to reconfigure the radio medium gives us another dimension of freedom to design wireless systems. Note that all conventional wireless optimization/designs take channel/radio environment as an input to adapt with, instead of actively ``tailoring/optimizing" it \cite{zheng2020fast}.

With multiple IRSs being launched (e.g., coated on various buildings) to support various group of users, it is critical to jointly optimize the phase-shifts and/or amplitudes of all IRSs to mitigate the interference as well as leverage the secondary reflections amongst IRSs. However, due to the multi-user interference, this type of problem is widely known to be non-convex. The problem even becomes more challenging when we need to jointly optimize not only the phase-shifts amongst all IRSs but also the beamforming vectors at the BS for other users. Multiple IRSs systems have recently been investigated, e,g., \cite{YangVPoor2020}-\cite{zheng2020doubleirs}. In \cite{YangVPoor2020}, the authors studied a resource allocation problem for a downlink wireless communication network with multiple distributed IRSs to maximize the system energy efficiency with maximum transmit power constraint and minimum rate requirements. In \cite{LiFangGao2019}, the authors considered a problem of maximizing the minimum SINR amongst users by jointly optimizing the transmit precoding vector at the BS and phase-shifts at IRSs for a downlink multiple IRSs system. In \cite{YangCosta2020}, the outage probability and average sumrate for a single source, single destination and multiple IRSs system were studied where the best IRS is selected to be active at a time. Note that all aforementioned works did not study the interference management nor leverage the cooperation amongst IRSs via the ``secondary" reflection. In the latest work \cite{zheng2020doubleirs}, the authors considered the beamforming design for the case with two IRSs in which one is placed close to the BS and the other is placed close to users. The authors show that there is a significant gain as compared to the case with a single IRS. Note that with only two IRSs that both aim to support a single group of users, the problem in \cite{zheng2020doubleirs} also does not account for the interference amongst IRSs as well as different groups of users. In practice, especially in a dense urban environment where shadowing and fading are severe (e.g., caused by multiple buildings/structures), multiple IRSs can be deployed within one cell to support various groups of users whose direct links to the BS are weak or blocked.  

Given the above, this work considers the uplink of multiple users that are grouped and supported by multiple IRSs to a multi-antenna BS, as shown in Figure \ref{fig:multiIRS}. Due to the close proximity amongst these groups of users, the IRSs can be deployed close to each other. It is clear that such a scenario requires us to deal with the interference amongst IRSs. However, we observe that we can also leverage the ``secondary" reflection amongst IRSs to enhance the signal reception at the BS. Specifically, the signal from a given user not only reach the BS by reflecting onto the dedicated IRS for the user's group (referred to as the ``primary" reflection) but also can traverse to and reflect on nearby IRSs before reaching the BS (referred to as the ``secondary" reflection). As such, optimizing/tuning the phase-shifts at one IRS \cite{WuZhang2019,YouZhengZhang2020,HuangZappone2019} or multiple IRSs but for a single group of users (e.g., \cite{YangVPoor2020}-\cite{zheng2020doubleirs})
is inapplicable. Note that these secondary reflections can either contribute to the (desired) signal detection (friends) for a given user or come as unwanted interference (foes) for other users at the BS. That require us to jointly optimize the IRS elements for \emph{all IRSs} and beamforming vectors for \emph{all users}. Note that results reported in \cite{zheng2020doubleirs} is a special case of the aforementioned problem by setting the number of IRSs to two and assign no user to the second IRS. 
Our simulations show that by properly managing interference and leveraging the secondary reflections amongst IRSs, there is a great benefit of deploying more IRSs to support different groups of users to achieve a higher rate per user.  
For example, we observed a 1.5 bps/Hz or 2.4 bps/Hz gain for a 2 or 4 IRSs system as compared to a single IRS system, respectively, or a 2.6 bps/Hz gain when exploiting the secondary reflection for a system with 6 IRSs.


The rest of this paper is organized as follows. In Section II, the system model and the problem formulation for multi-IRS system are presented. In Section III, a sub-optimal solution are presented. In Section IV, numerical results are discussed, and conclusions are drawn in Section V.

\section{System Model}
\subsection{Introduction and Notation}

Let's consider a system with an $N$-antenna BS and $K$ users. Due to the blockages, users are partitioned/clustered into $L$ groups whose direct radio links to the BS are blocked. That is often the case in dense CBD/urban areas with high-rise buildings and structures. $L$ IRSs are then used to support these users, as depicted in Figure \ref{fig:multiIRS}. We assume that each IRS $l$th is comprised of $M_l$ sub-surfaces and 
serves $K_l$ single-antenna users in the $l$-th group $l=1,2,\ldots,L$. As aforementioned, unlike the special case in \cite{zheng2020doubleirs} which only considered two IRSs, one group of users and single directional reflection between two IRSs, we assume each IRS is placed close to a group of users and signals reflected between any two IRSs are bi-directional. Specifically, the signal from the $k$-th user in the $l$-th group (denoted as the ($l,k$)-th user for short) reaches the BS by propagating from the user to IRS-$l$ to the BS and also via the reflection from IRS-$l$ to IRS-$l'$ to the BS (with $l'\neq l$).  Similarly the signal of the $(l',i)$ user can also reach the BS by traveling from the $(l',i)$-th user to the IRS-$l'$ to IRS-$l$ to the BS. For that, we can leverage that fact that signal can travel between the IRS-$l$ and the IRS-$l'$ in both direction (yet the signals are originated from different users) or the cooperation amongst IRSs. Given this observation, the phase-shift design for the IRSs and the receive beamforming at the BS need to take in to account both the primary IRS reflections and the secondary reflections amongst IRSs.  
\begin{figure}[htbp]
  \centering
  \includegraphics[width=90mm,height=63mm,scale=1.0]{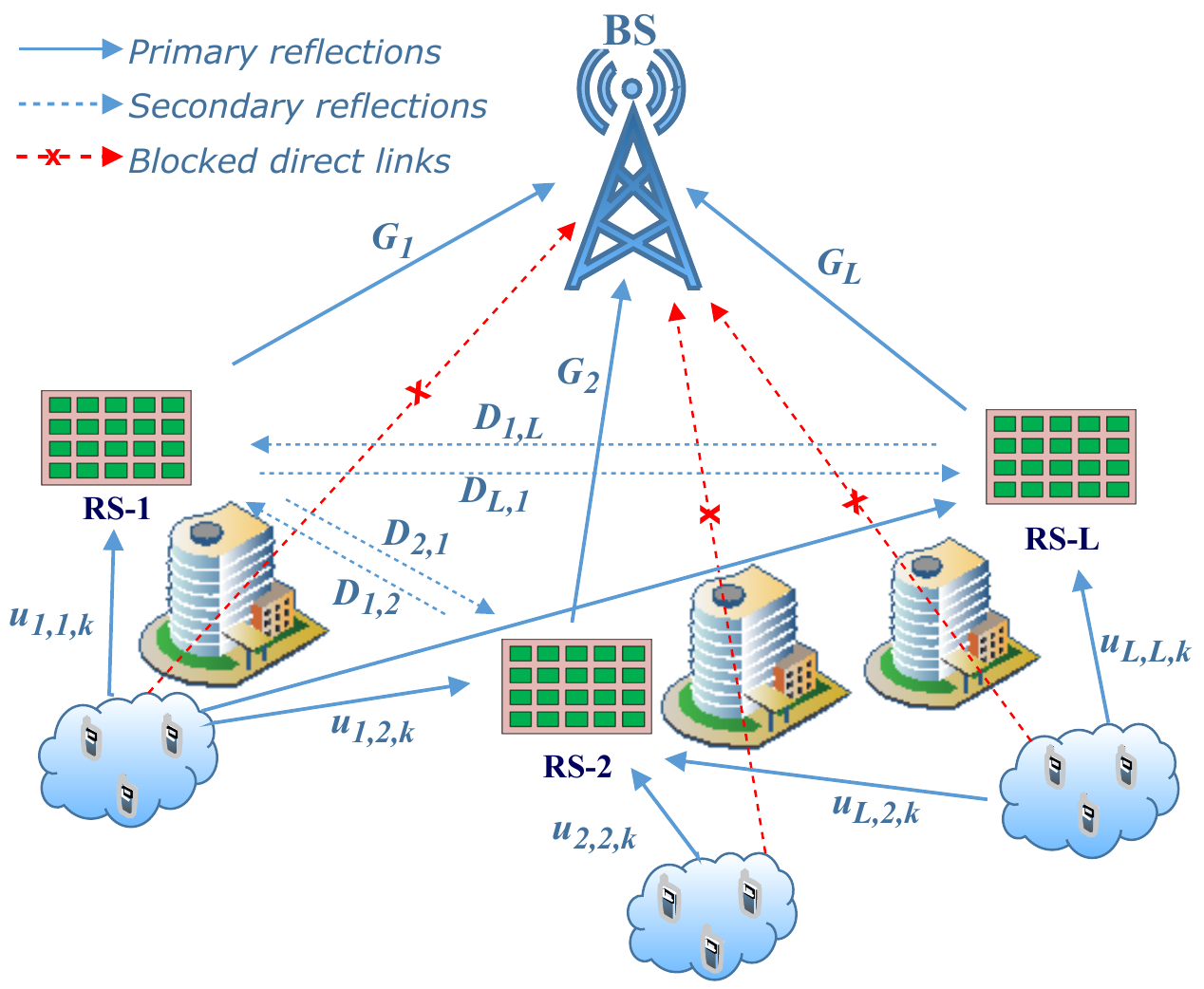}
  \caption{Multi-IRS assisting multi-groups of user MIMO communication system.}
  \label{fig:multiIRS}
\end{figure}

Let $\boldsymbol{u}_{l',l,k}\in \mathbb{C}^{M_{l'}\times1}$,  $\boldsymbol{D}_{l,l'}\in \mathbb{C}^{M_l\times M_{l'}}$, $\boldsymbol{G}_l\in \mathbb{C}^{N\times M_l}$ denote baseband equivalent channel matrices for the user $k$ in the group $l$ to the IRS-$l'$, the IRS-$l'$ to the IRS-$l$ (for $l'\neq l$), and IRS-$l$ to the BS links, respectively, with $l,l'=1,2,\ldots,L$ and $k=1,2,\ldots,K_l$. Let $\boldsymbol{\theta}_l \in \mathbb{C}^{M_{l}\times1}$ denote the phase-shifts vector of the IRS-$l$ and $K=\sum_{l=1}^L K_l$ denote the total number of users.

In general, the signal from the $(l,k)$-th user can reach the BS by one of the following paths via at least one IRS: (i) single IRS reflection from the $(l,k)$-th user onto the $l'$-th IRS to the BS, (ii) double IRSs reflections from $(l,k)$-th user to the $l$-th IRS, then to the $l'$-th IRS to the BS (where $l\neq l'$), (iii) triple or more IRSs reflections, e.g., from $(l,k)$-th user to the $l$-th IRS, then to the $l'$-th IRS to the $l''$-th IRS to the BS (where $l$, $l'$ and $l''$ are pair-wise different). In practice, if a signal travels through many IRSs, due to much larger effective pathloss, the signal arrival at the BS become very weak to be considered. Hence, in this paper, we ignore the type (iii), the tripple and more IRSs reflection mentioned above, i.e., we ignore any signal that reflected through three or more different IRSs. Also, for the type (ii), if the signal of an user is reflected from two IRSs that are far away from each other, the secondary reflection in this case is also too weak to be considered.

The signals from the $(l,k)$-th user arrive at the BS are the $L$ reflection paths via each of the $L$ IRSs (type (i) single reflection) which we called \emph{primary} refection, and the $L-1$ paths reflection from the $l$-th IRS to the remaining $l'\neq l$ IRS (type (ii) double/secondary reflection). Again, we assume the direct links from users to the BS do not exist (or too weak to be considered). Thus, the effective channel from the $(l,k)$-th user can be written as follows
\begin{align*}
    \boldsymbol{h}_{l,k} =& \sum_{l'=1}^L \boldsymbol{G}_{l'}\boldsymbol{\Phi}_{l'}  \boldsymbol{u}_{l',l,k} + \sum_{{l'=1},{l'\neq l}}^L \boldsymbol{G}_{l'}\boldsymbol{\Phi}_{l'} \boldsymbol{D}_{l',l} \boldsymbol{\Phi}_l \boldsymbol{u}_{l,l,k} \numberthis \label{eqn:chan_mk}
\end{align*}
for $l=1,2,\ldots,L$, and $k=1,2,\ldots,K_l$, where $\boldsymbol{\Phi}_l=\mathrm{diag}(\boldsymbol{\theta}_l)$ denotes the diagonal reflection matrix of IRS-$l$. 

Let denote $\boldsymbol{R}_{l',l,k}=\boldsymbol{G}_{l'} \mathrm{diag}(\boldsymbol{u}_{l',l,k})$ be the cascaded channel matrix from the $(l,k)$-th user to the $l'$-th IRS to the BS (without the phase shifts at the $l'$-th IRS). Also, let denote $\tilde{\bm{D}}_{l',l,k}=\left[\tilde{\bm{d}}_{l',l,k,1},\ldots,\tilde{\bm{d}}_{l',l,k,M_{l'}}\right]\triangleq \bm{D}_{l',l}\mathrm{diag}(\bm{u}_{l,l,k})$, we can rewrite (\ref{eqn:chan_mk}) as follow
\begin{align*}
    \boldsymbol{h}_{l,k} =&  \sum_{l'=1}^L\boldsymbol{R}_{l',l,k}\boldsymbol{\theta}_{l'} + \sum_{{l'=1},{l'\neq l}}^L \boldsymbol{G}_{l'}\boldsymbol{\Phi}_{l'} \tilde{\bm{D}}_{l',l,k}\boldsymbol{\theta}_l \\
    =& \sum_{l'=1}^L\boldsymbol{R}_{l',l,k}\boldsymbol{\theta}_{l'} + \sum_{{l'=1},{l'\neq l}}^L \sum_{j=1}^{M_{l}} \boldsymbol{Q}_{l',l,k,j} \boldsymbol{\theta}_{l'} \theta_{l,j} \numberthis\label{eqn:channelmk}
\end{align*} 
\normalsize
where $\boldsymbol{Q}_{l',l,k,j}\triangleq\boldsymbol{G}_{l'} \mathrm{diag}(\tilde{\bm{d}}_{l',l,k,j})$ denotes the cascaded channel matrix from the $(l,k)$-th user to the $j$-th subsurface element of the IRS-$l$ to the IRS-$l'$ to the BS without the phase shifts at the IRS-$l$ and the IRS-$l'$, (where $l=1,2,\ldots,L$, $j=1,2,\ldots,M_l$ and $k=1,2,\ldots,K_l$). From equation (\ref{eqn:channelmk}), we can see that it suffices to estimate $\{\boldsymbol{R}_{l',l,k}\}$ and $\{\boldsymbol{Q}_{l',l,k,j}\}$ for jointly designing the passive beamforming coefficients $\{\boldsymbol{\theta}_l\}$ in the multiple IRS cooperative system \cite{zheng2020doubleirs, zheng2019chanest}. In this paper, we assume all the cascaded channel matrices $\{\boldsymbol{R}_{l',l,k}\}$ and $\{\boldsymbol{Q}_{l',l,k,j}\}$ are known at the BS.

During the uplink data transmission, the received signal at the BS is given by
\begin{align*}
    \boldsymbol{y} =& \sum_{l=1}^L\sum_{k=1}^{K_l} \boldsymbol{h}_{l,k} s_{l,k} + \boldsymbol{n} \numberthis \label{aln:mrxulsignal2}\\
    =& \sum_{l=1}^L \sum_{k=1}^{K_l} \left(\sum_{l'=1}^L \boldsymbol{R}_{l',l,k} \boldsymbol{\theta}_{l'} + \right. \sum_{{}^{l'=1}_{l'\neq l}}^L \left.\sum_{j=1}^{M_{l}} \boldsymbol{Q}_{l',l,k,j} \boldsymbol{\theta}_{l'} \theta_{l,j}\right) s_{l,k} + \boldsymbol{n} 
\end{align*}
\normalsize
where $s_{l,k}$ is the $(l,k)$-th user's transmitted data symbol with the transmit power of $P_{l,k}$, and $\boldsymbol{n}\sim\mathcal{N}_c(0, \sigma^2\boldsymbol{I})$ is the additive white Gaussian noise (AWGN) vector at the BS with $\sigma^2$ being the equivalent noise power. At the BS, a linear receive beamforming vector $\boldsymbol{w}_{l,k} \in \mathbb{C}^{N\times 1}$ is applied to decode each $s_{l,k}$, written as
\begin{align*}
    \tilde{y}_{l,k}=~&\boldsymbol{w}_{l,k}^H\boldsymbol{y}\\ =~&\boldsymbol{w}_{l,k}^H \sum_{l'=1}^L\sum_{k'=1}^{K_l} \left(\sum_{l''=1}^L \boldsymbol{R}_{l'',l',k'}\boldsymbol{\theta}_{l''} + \right. \sum_{l''\neq l'}^L \sum_{j=1}^{M_{l'}} \boldsymbol{Q}_{l'',l',k',j} \boldsymbol{\theta}_{l''} \theta_{l',j}\Bigg) s_{l',k'} + \boldsymbol{w}_{l,k}^H\boldsymbol{n} \numberthis \label{eqn:mrxulsignal_k}
\end{align*}
\normalsize
Therefore, the signal to interference and noise ratio (SINR) for decoding the information from the $(l,k)$-th user is given by (\ref{eqn:gamma_mk}).

\begin{align}
\label{eqn:gamma_mk}
\gamma_{l,k} = \frac{P_{l,k}\left| \boldsymbol{w}_{l,k}^H \left( \sum\limits_{l'=1}^L\boldsymbol{R}_{l',l,k}\boldsymbol{\theta}_{l'} + \sum\limits_{l'\neq l}^L \sum\limits_{j=1}^{M_l} \boldsymbol{Q}_{l',l,k,j} \boldsymbol{\theta}_{l'} \theta_{l,j} \right)  \right|^2 }
{\sum\limits_{{(l',k')}{\neq(l,k)}} P_{l',k'} \left|\boldsymbol{w}_{l,k}^H   \left(\sum\limits_{l''=1}^L\boldsymbol{R}_{l'',l',k'}\boldsymbol{\theta}_{l''} + \sum\limits_{l''\neq l'}^L \sum\limits_{j=1}^{M_{l'}} \boldsymbol{Q}_{l'',l',k',j} \boldsymbol{\theta}_{l''} \theta_{l',j}\right) \right|^2 + \sigma^2\boldsymbol{w}_{l,k}^H \boldsymbol{w}_{l,k} }
\end{align}

\subsection{Uplink Joint Multi-User MIMO Receive Beamforming and Multi-IRS Optimization Problem}

Given the SINR in (\ref{eqn:gamma_mk}) for all $(l,k)$ users, we can optimize $\boldsymbol{w}_{l,k}$ and $\{\boldsymbol{\theta}_l\}$'s to simultaneously mitigate the IRS interference and leverage the secondary reflections amongst IRSs at the BS. In this paper, similar to \cite{zheng2020doubleirs}, let's first aim to maximize the minimum achievable rate among all users, that is
\setcounter{equation}{5}
\begin{align*}
\mathrm{(P1): }~ & \numberthis \label{meqn:optP1} \max_{\{\boldsymbol{w}_{l,k}\},\{\boldsymbol{\theta}_l\}} ~~~\min_{\{l,k\}} ~~ \log_2(1+\gamma_{l,k}) \\
    ~~& \mathrm{s.t.} ~~~ |\theta_{l,j}|=1,~ \forall l=1,2,\ldots,L,~ j=1,2,\ldots,M_l,
\end{align*}
where $\gamma_{l,k}$ is given in (\ref{eqn:gamma_mk}), and the constraints $|\theta_{l,j}|=1$ is due to the IRS infection type assumption.

\begin{lemma}
    \label{theorem:nonConvexP1}
    The problem (P1) in (\ref{meqn:optP1}) is non-convex.
\end{lemma}

\begin{proof}
It can be observed due to the probem's non-convex feasible region (because $|\theta_{l,j}|=1$ are non-convex).
\end{proof}


In the sequel we will use the alternating optimization (AO) algorithm to find a sub-optimal solution to the non-convex problem (P1).

\section{AO Algorithm Based Solution}
In this section, we use the AO algorithm to solve for (\ref{meqn:optP1}) to obtain a sub-optimal solution. First we assume all the IRSs phase-shifts $\{\boldsymbol{\theta}_l\}_{l=1}^L$ are known and fixed, we then optimize the receive beamforming $\{\boldsymbol{w}_{l,k}\}$ for all users $(l,k)$. Then, for each $l=1,2,\ldots,L$, the $\boldsymbol{\theta}_l$ is optimized while the $\{\boldsymbol{w}_{l,k}\}$'s and all other $\{\boldsymbol{\theta}_{l'}\}_{l'\neq l}^L$ are fixed. Note that although the AO method is also used in \cite{zheng2020doubleirs}, in our case, solving the (P1) is not straightforward due to the complexity introduced by multiple IRSs, bi-directional reflection among IRSs, and multiple groups of users.

\subsection{Optimize $\{\boldsymbol{w}_{l,k}\}$'s for fixed $\{\boldsymbol{\theta}_l\}_{l= 1}^L$}

For fixed $\{\boldsymbol{\theta}_l\}_{l= 1}^L$, the effective channel of each user $\boldsymbol{h}_{l,k}$ in (\ref{eqn:channelmk}) is fixed and thus problem (P1) is reduced to $K$ sub-problems, each of which is to maximize the SINR of the $(l,k)$-th user as given in (\ref{eqn:gamma_mk}) and can be formulated as

\begin{align*}
\mathrm{(P2):}~ \numberthis \label{eqn:optP4} \max_{\boldsymbol{w}_{l,k}} ~ \frac{P_{l,k}\boldsymbol{w}_{l,k}^H\boldsymbol{h}_{l,k}\boldsymbol{h}_{l,k}^H\boldsymbol{w}_{l,k}}{\boldsymbol{w}_{l,k}^H\Bigg(\sum\limits_{(l',k')\neq(l,k)} P_{l',k'}\boldsymbol{h}_{l',k'}\boldsymbol{h}_{l',k'}^H + \sigma^2\boldsymbol{I}\Bigg)\boldsymbol{w}_{l,k}}
\end{align*}
\normalsize
It can be shown that (P2) is a convex optimization problem, and in fact it has a closed-form and optimal solution \cite{MKayBook}. To simplify the notation, let's denote $\boldsymbol{H}=[\boldsymbol{H}_{1},\ldots,\boldsymbol{H}_{L}]\in\mathbb{C}^{N\times K}$ and $\boldsymbol{W}=[\boldsymbol{W}_{1},\ldots,\boldsymbol{W}_{L}]\in\mathbb{C}^{N\times K}$, where $\boldsymbol{H_l}=[\boldsymbol{h}_{l,1},\ldots,\boldsymbol{h}_{l,K_l}]\in\mathbb{C}^{N\times K_l}$ and $\boldsymbol{W_l}=[\boldsymbol{w}_{l,1},\ldots,\boldsymbol{w}_{l,K_l}]\in\mathbb{C}^{N\times K}$, denote the effective user-to-BS channel matrix and the receive beamforming matrix applied at the BS, respectively. The problem (P2) has the optimal minimum mean squared error (MMSE) solution as below \cite{MKayBook}
\begin{equation}
    \label{eqn:mmse_sol}
    \boldsymbol{W}_{MMSE} = (\boldsymbol{HPPH}^H+\sigma^2\boldsymbol{I})^{-1}\boldsymbol{HP}
\end{equation}
where $\boldsymbol{P}=\diag(\left[\sqrt{P_{1,1}},\sqrt{P_{1,2}},\ldots,\sqrt{P_{L,K_L}}\right])$ is a diagonal transmit power matrix for all $K$ users. In practice, the equivalent noise power, $\sigma^2$, needs to be estimated; the solution to the problem (P2) can often be simplified to a sub-optimal zero-forcing (ZF) solution where the noise term $\sigma^2\boldsymbol{I}$ in (\ref{eqn:mmse_sol}) is omitted \cite{MKayBook}. 

\subsection{Optimize $\{\boldsymbol{\theta}_l\}_{l=1}^L$ for fixed $\{\boldsymbol{w}_{l,k}\}$'s}

For fixed $\{\boldsymbol{w}_{l,k}\}$'s, the problem (P1) in (\ref{meqn:optP1}) is equivalent to the following optimization problem in which we optimized the SINR $\gamma_{l,k}$ directly
\begin{align*}
\mathrm{(P3)} ~~~ & \numberthis \label{meqn:optP3} \max_{\{\boldsymbol{\theta}_l\}, \delta} ~~~\delta, \mathrm{~~~~s.t.}~~\gamma_{l,k} \geq \delta, ~ \mathrm{for all} ~(l,k), \\
    ~~& ~~~~|\theta_{l,j}|=1,~ \mathrm{for all} ~l=1,2,\ldots,L, ~ j=1,2,\ldots,M_l
\end{align*}

\begin{lemma}
    \label{theorem:nonConvexP3}
    The problem (P3) in (\ref{meqn:optP3}) is non-convex.
\end{lemma}

\begin{proof}
	The problem (P3) is non-convex because its constraints  $|\theta_{l,j}|=1$ is not a convex set.
\end{proof}


We now can use the AO algorithm to solve (P3) by alternatively optimizing each $\boldsymbol{\theta}_l$ while considering all the other phase-shifts to be fixed as in the following subsection. 

\subsection{Optimize $\boldsymbol{\theta}_{\bar{l}}$ for fixed $\{\boldsymbol{\theta}_{l'}\}_{l'\neq \bar{l}}$ and fixed $\{\boldsymbol{w}_{l,k}\}$'s}
For fixed $\{\boldsymbol{\theta}_{l'}\}_{l'=1,l'\neq \bar{l}}^L$ and fixed $\{\boldsymbol{w}_{l,k}\}$'s, the problem (P3) is reduced to the following problem
\begin{align*}
\mathrm{(P3.1)} ~~~ & \numberthis \label{meqn:optP3.1} \max_{\boldsymbol{\theta}_{\bar{l}}, \delta} ~~~\delta,  \mathrm{~~~~s.t.}~~~ \gamma_{l,k} \geq \delta ~ \mathrm{for all} ~(l,k), \\
    ~~&  ~~~~~~~ |\theta_{\bar{l},j}|=1,~ \mathrm{for all} ~ j=1,2,\ldots,M_{\bar{l}}
\end{align*}
where $\bar{l}=1,2,\ldots,L$ is a predefined value. Note that this (P3.1) is also a non-convex problem due to the unity constraints on $|\theta_{\bar{l},j}|$ as mentioned in Lemmas  \ref{theorem:nonConvexP1} and \ref{theorem:nonConvexP3}. 

In the following, we first specify $\gamma_{l,k}$ as a function of $\boldsymbol{\theta}_{\bar{l}}$. There are two scenarios to be considered: i) the $(l,k)$-th user belongs to the $\bar{l}$-th IRS (i.e., $l=\bar{l}$) and ii) the $(l,k)$-th user does not belong to the $\bar{l}$-th IRS (i.e., $l\neq \bar{l}$).

\begin{theorem}
\label{theorem:GenericSINR}
The SINR of the $(l,k)$-th user at the BS can be written as a function of $\boldsymbol{\theta}_{\bar{l}}$ as follow.
\begin{align}
\label{meqn:rxsinr_mk_reduced}
\gamma_{l,k} = \frac{\left|\boldsymbol{q}^H_{\bar{l},l,k} \boldsymbol{\theta}_{\bar{l}} + \bar{q}_{\bar{l},l,k}\right|^2}
{\sum\limits_{(l',k')\neq (l,k)} \left| \boldsymbol{q}^H_{\bar{l},l',k'} \boldsymbol{\theta}_{\bar{l}} + \bar{q}_{\bar{l},l',k'}\right|^2 + \sigma_{l,k}^2}
\end{align}
where  $\sigma_{l,k}^2 \triangleq \sigma^2\boldsymbol{w}_{l,k}^H \boldsymbol{w}_{l,k}$,
\begin{align}
\label{eqn:def_q_lmk}
\boldsymbol{q}_{\bar{l},l,k} &\triangleq \left\{
    \begin{array}{ll}
         \sqrt{P_{l,k}}\boldsymbol{w}^H_{l,k}\left(\boldsymbol{R}_{\bar{l},l,k} + \boldsymbol{S}_{\bar{l},l,k}\right) &~\mathrm{for} ~ l = \bar{l}\\
         \sqrt{P_{l,k}}\boldsymbol{w}^H_{l,k}\left(\boldsymbol{R}_{\bar{l},l,k} + \boldsymbol{T}_{\bar{l},l,k}\right) &~\mathrm{for} ~ l \neq \bar{l}
    \end{array} 
    \right., \mathrm{ and}\\
\label{eqn:barq_lmk}
\bar{q}_{\bar{l},l,k} &\triangleq \left\{
    \begin{array}{ll}
         \sqrt{P_{l,k}}\boldsymbol{w}^H_{l,k}\boldsymbol{U}_{\bar{l},l,k} &\mathrm{for} ~ l = \bar{l}\\
         \sqrt{P_{l,k}}\boldsymbol{w}^H_{l,k}\big(\boldsymbol{U}_{\bar{l},l,k} +\boldsymbol{S}_{\bar{l},l,k}\boldsymbol{\theta}_{l} \big) &\mathrm{for} ~ l \neq \bar{l}
    \end{array} 
    \right.
\end{align}
\normalsize
with $\boldsymbol{S}_{\bar{l},l,k}$, $\boldsymbol{T}_{\bar{l},l,k}$ and $\boldsymbol{U}_{\bar{l},l,k}$ being independent of $\boldsymbol{\theta}_{\bar{l}}$ as defined in (\ref{eqn:definedS}) and (\ref{eqn:definedTU}). 

\end{theorem}

\begin{proof}
	See Appendix~\ref{App:SINR_proof}.
\end{proof}
We note that (\ref{meqn:rxsinr_mk_reduced}) has the same formula for all $(l,k)$ but the underlying terms are different for users that belong to the $\bar{l}$-th IRS and those do not, as shown in (\ref{eqn:def_q_lmk}) and (\ref{eqn:barq_lmk}). Hence, when only $\boldsymbol{\theta}_{\bar{l}}$ is varying, the problem (P3.1) in (\ref{meqn:optP3.1}) can be written as the following optimization problem.

\begin{align*}
&\mathrm{(P3.2)} ~ \max_{\boldsymbol{\theta}_{\bar{l}}, \delta} ~~~\delta \numberthis \label{meqn:optP3.2}\\
&~\mathrm{s.t.}~ \left|\boldsymbol{q}^H_{\bar{l},l,k} \boldsymbol{\theta}_{\bar{l}} + \bar{q}_{\bar{l},l,k}\right|^2 \geq 
    \delta \sum\limits_{(l',k')\neq (l,k)} \left| \boldsymbol{q}^H_{\bar{l},l',k'} \boldsymbol{\theta}_{\bar{l}} + \bar{q}_{\bar{l},l',k'}\right|^2 + \delta \sigma_{l,k}^2  \\
& ~~~~~~ |\theta_{\bar{l},j}|=1,~ \forall j=1,2,\ldots,M_{\bar{l}}, ~\mathrm{and ~for ~all} ~(l,k)
\end{align*}
\normalsize
We can rewrite

\begin{align*}
    &\left| \boldsymbol{q}^H_{\bar{l},l,k} \boldsymbol{\theta}_{\bar{l}} + \bar{q}_{\bar{l},l,k}\right|^2 
    =~ \Tr\left( \boldsymbol{B}_{\bar{l},l,k}\tilde{\boldsymbol{\theta}}_{\bar{l}}\tilde{\boldsymbol{\theta}}^H_{\bar{l}} \right) + \left|\bar{q}_{\bar{l},l,k}\right|^2, \mathrm{ where}\\
    &\boldsymbol{B}_{\bar{l},l,k} \triangleq \left[
    \begin{array}{cc}
        \boldsymbol{q}_{\bar{l},l,k}\boldsymbol{q}^H_{\bar{l},l,k} & \bar{q}_{\bar{l},l,k}\boldsymbol{q}_{\bar{l},l,k}  \\
        \bar{q}^H_{\bar{l},l,k}\boldsymbol{q}^H_{\bar{l},l,k} & 0 
    \end{array}
    \right], ~~
    \tilde{\boldsymbol{\theta}}_{\bar{l}} \triangleq \left[
    \begin{array}{c}
        \boldsymbol{\theta}_{\bar{l}}   \\
        t 
    \end{array}
    \right],
\end{align*}
\normalsize
with $t$ is an auxiliary variable ($t=1$ to be exact but we are going to relax this condition). Let define $\Psi_l\triangleq\tilde{\boldsymbol{\theta}}_l\tilde{\boldsymbol{\theta}}^H_l$, we have $\Psi_l \succeq 0$ and $\rank(\Psi_l)=1$. Because the rank-one constraint, $\rank(\Psi_l)=1$, is non-convex, we relax this constraint. For that, the problem (\ref{meqn:optP3.2}) is rewritten as

\begin{align*}
&\mathrm{(P3.3)} ~~  \numberthis \label{meqn:optP3.3} \max_{\Psi_{\bar{l}}, \delta} ~~~\delta, \quad \mathrm{s.t.}~ \Tr(\boldsymbol{B}_{\bar{l},l,k}\Psi_{\bar{l}}) + \left|\bar{q}_{\bar{l},l,k}\right|^2 \geq \quad\\
    \quad& ~\delta \sum\limits_{(l',k')\neq (l,k)} \bigg(\Tr(\boldsymbol{B}_{\bar{l},l',k'}\Psi_{\bar{l}}) ~+~ \left|\bar{q}_{\bar{l},l',k'}\right|^2\bigg) +\delta\sigma_{l,k}^2 , \\
    \quad&~\Psi_{\bar{l}} \succeq 0, ~~ [\Psi_{\bar{l}}]_{jj}=1,~~\forall~ j=1,\ldots,M_{\bar{l}}+1, \mathrm{ and ~for ~all ~} (l,k).
\end{align*}
\normalsize
For a fixed $\delta$, (P3.3) is a convex semidefinite program (SDP) problem and reduced to a feasibility-check problem \cite{zheng2020doubleirs,BoydBook}, which can be optimally solved by standard convex optimization solvers \cite{CVXBoyd}. Therefore, (P3.3) can be efficiently solved by the bisection method; that is, we do a binary search on $\delta$ that has a feasible solution $\Psi_{\bar{l}}$. Once we obtain a maximum $\delta$ (up to a certain numerical accuracy) with a $\Psi_{\bar{l}}$ solution, we will use the Gaussian randomization search to obtain a solution for $\boldsymbol{\theta}_{\bar{l}}$ \cite{ruizhang2019gaussian}. Our simulation show that the numerical solution $\Psi_{\bar{l}}$ is rank deficient (i.e., the max eigen value is much higher than the rest), and we can almost always obtain a numerical solution $\boldsymbol{\theta}_{\bar{l}}$ such that the max-min achievable rate is within a pre-defined error, e.g., 0.1\%, difference from the best $\delta_{opt}$ of the problem (P3.3). To make sure that the algorithm converges we add a heuristic check if the new $\boldsymbol{\theta}_{\bar{l}}$ reduces the original objective in (P1) then it will not be updated such that the objective function in (P1) is always non-decreasing.

\begin{algorithm}[h]
	\caption{An iterative algorithm to solve (P1) in \eqref{meqn:optP1}.}
	\label{algorithm1}
	\begin{algorithmic}[1]
		\State \textbf{Input:} The previous output $\{\bm{w}_{l,k}^{(n-1)}, \bm{\theta}_l^{(n-1)}\}$.
		\State \textbf{Initialize:} $n = 1$, $\{\bm{w}_{l,k}^{(0)}, \bm{\theta}_l^{(0)}\}$, max number of iterations $I_1\ge1$, tolerance $\xi > 0$ and $\epsilon>0$.
		\State \textbf{Compute:} the (minimum) achievable rate $\gamma_{\min}^{(0)}=\min\limits_{l,k}\gamma_{l,k}\!\big( \bm{w}_{l,k}^{(0)}, \bm{\theta}_1^{(0)}, \bm{\theta}_2^{(0)},\ldots, \bm{\theta}_L^{(0)} \big)$.
		\State \textbf{Repeat:}
		\State \quad {Obtain $\bm{w}_{l,k}^{(n)}$ from $\{\bm{\theta}_l^{(n-1)}\}$ by solving \eqref{eqn:mmse_sol};}
		\State \quad {\textbf{Foreach: $\bar{l}=1,2,\ldots,L$:}}
		\State \quad\quad {Obtain $\bm{\theta}_{\bar{l}}^{(n)}$ by solving \eqref{meqn:optP3.3} with $\bm{w}_{l,k}=\bm{w}_{l,k}^{(n)}$, $\bm{\theta}_l=\bm{\theta}_l^{(n)}$ for $l<\bar{l}$, and $\bm{\theta}_l=\bm{\theta}_l^{(n-1)}$ for $l>\bar{l}$, and tolerance $\epsilon$;} 
		\State \quad \textbf{Compute:} the (minimum) achievable rate $\gamma_{\min}^{(n)}=\min\limits_{l,k}\gamma_{l,k}\!\big( \bm{w}_{l,k}^{(n)}, \bm{\theta}_1^{(n)}, \bm{\theta}_2^{(n)},\ldots, \bm{\theta}_L^{(n)} \big)$;
		\State \quad {\textbf{If:}} 
		\State{\quad\quad $\left| \gamma_{\min}^{(n)} - \gamma_{\min}^{(n-1)} \right| \!<\! {\xi}$  ~or~  $n\ge I_1$;} 
		\State \quad {\textbf{Then:}} 
		\State \quad\quad {Set $\{\bm{w}_{l,k}^*, \bm{\theta}_l^*\} = \{\bm{w}_{l,k}^{(n)}, \bm{\theta}_l^{(n)}\}$ and terminate.}
		\State{\quad \textbf{Otherwise:}} 
		\State \quad \quad {Update $n \leftarrow  n + 1$} and continue.
		\State {\textbf{Output:}} The optimal solution $\bm{\hat \chi^*} = \{\bm{w}_{l,k}^*, \bm{\theta}_l^*\}$.
	\end{algorithmic}
\end{algorithm}
\normalsize

\subsection{Iterative AO Algorithm for Max-Min Achievable Rate}

The proposed iterative algorithm is summarized in the \textbf{Algorithm~\ref{algorithm1}}. Its convergence is provided in the following Lemma.


\begin{lemma}
	\label{lemma:convergence_and_local_opt}
	The objective values $\gamma_{\min}^{(n)}$ as shown in the Algorithm~\ref{algorithm1} is monotonically non-decreasing and thus the algorithm is guaranteed to be converged.
\end{lemma}

\begin{proof}
	It can be shown that at each step the objective values is monotonically non-decreasing, also due to the limited transmit power and thermal noise, the SINR is bounded above by the SNR; hence it follows the Lemma \ref{lemma:convergence_and_local_opt}. The detailed proof is omitted due to space limitation.
\end{proof}

\subsection{Complexity Analysis}
The complexity of solving optimization problem (\ref{meqn:optP1}) is the combined complexity of solving the beamforming weights as shown in (\ref{eqn:mmse_sol}) whose complexity is mainly from the computation of inverse of an $N\times N$ matrix, hence $\mathcal{O}(N^3)$. Similar to the analysis in \cite{YangVPoor2020}, the complexity of solving the $l$-th SDP optimization problem in (\ref{meqn:optP3.3}) with the bisection method, which has $M^2_l$ variables, is  $\mathcal{O}(M^{4.5}_l \log(1/\epsilon))$. Here $\epsilon$, as shown in the Algorithm 1, is the accuracy of the bisection search. The complexity for $L$ sub-problems is $\mathcal{O}((M^{4.5}_1+\ldots+M^{4.5}_L)\log(1/\epsilon))$. Hence, the total complexity of Algorithm \ref{algorithm1} is $\mathcal{O}(I_1((M^{4.5}_1+\ldots+M^{4.5}_L)\log(1/\epsilon) + N^3)$, or $\mathcal{O}(I_1(LM^{4.5}_1\log(1/\epsilon) + N^3)$ when $M_l=M_1$ for all $l$.

It is worth noting that at each step, the solution $\bm{\theta}^*_l$ results an objective value closed to the optimal value up to a tolerance $\epsilon$. 
Our simulations below show that the Algorithm 1 converges after about 50 iterations.


\section{Numerical results}
\label{sec:Numerical}
In this section, we consider an uplink OFDM system in which each user is allocated with $B=180$KHz bandwidth. Thus, the noise power at the BS is assumed to be $\sigma^2=-174+10\log_{10}(B)$ dBm. We assume LOS channel models between BS and IRSs and between any two IRSs with pathloss exponent $\alpha_1=2.2$ and Rician factor of 5dB, and non-LOS channel model between the IRSs and users with pathloss exponent $\alpha_1=3.0$. Also, the elements of BS and RISs have 5 dBi gain, while users have a single antenna with 0 dBi gain \cite{Bjornson2020}. We assume that the BS has 16 antennas while varying the number of IRSs as well as the number of elements per IRS from 32 to 128. Similar to the simulation setup and assumption in \cite{QNadeem2020,ZhangUAV2021}, we consider a small-cell situation as shown in Figure \ref{fig:IRSLocation}. Particularly, we assume the distance between BS and the nearest IRS is about 60 meters, multiple IRSs are serving multiple users groups in an area of 40$\times$20 squared meters. The locations of users are assumed to be uniformly distributed. In this figure, there are 4 IRSs in which two adjacent IRSs are 10 meters apart. Each IRS serves 3 users in an area of 10$\times$20 squared meters. We assumed the heights of BS, IRSs, and users are 25, 30 and 1.5 meters, respectively. The transmit power at each user is assumed to be the same. 

In Figure \ref{fig:IRSconfig}, we consider the achievable rate for the lowest-rate user with various configurations: (i) 4 IRSs with 3 users, (ii) 2 IRSs with 6 users and (iii) 1 IRS with 12 users; that is, there are total 12 users. The locations for 4 IRSs scenario are shown in Figure \ref{fig:IRSLocation}. When there are 2 IRSs, their locations are at $(x_{irs},y_{irs})=(60,\pm10)$ meters, and for a single IRS scenario, its location is at $(60,10)$ meters. In all these configurations, each IRS has 64 elements and the BS has 16 antennas. We observe that with the same number of users and receive antennas, the more IRSs are used, the higher the achievable rate. At the transmit power at 30 dBm, for instance, we observe a 1.5 bps/Hz or 2.4 bps/Hz gain for a 2 or 4 IRSs system as compared to a single IRS system, respectively. This is because of higher received signal power at the receiver antennas due to the reflections from the additional IRSs. 
In this figure, we also plot the average achievable rate among all users and observe that the average achievable rate follows a similar trend, and is not much different from the lowest rate. Hence, in the following, we will consider the performance for the lowest rate user only.

\begin{figure}[t]
  \centering
  \includegraphics[width=85mm,scale=1.0]{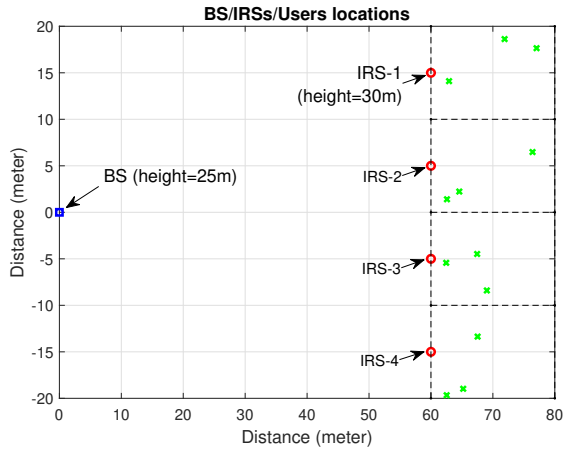}
  \caption{A snapshot of BS, IRSs, and users locations}
  \label{fig:IRSLocation}
\end{figure}

\begin{figure}[t]
  \centering
  \includegraphics[width=80mm,scale=1.0]{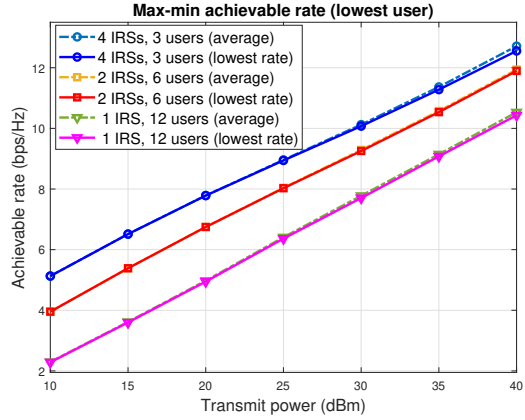}
  \caption{Achievable rate with various IRS configurations}
  \label{fig:IRSconfig}
\end{figure}

\begin{figure}[t]
  \centering
  \includegraphics[width=80mm,scale=1.0]{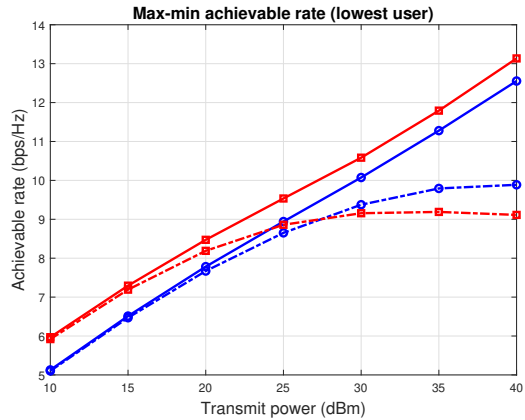}
  \caption{Leveraging secondary reflections amongst IRSs}
  \label{fig:skipSecondary}
\end{figure}

\begin{figure}[t]
  \centering
  \includegraphics[width=80mm,scale=1.0]{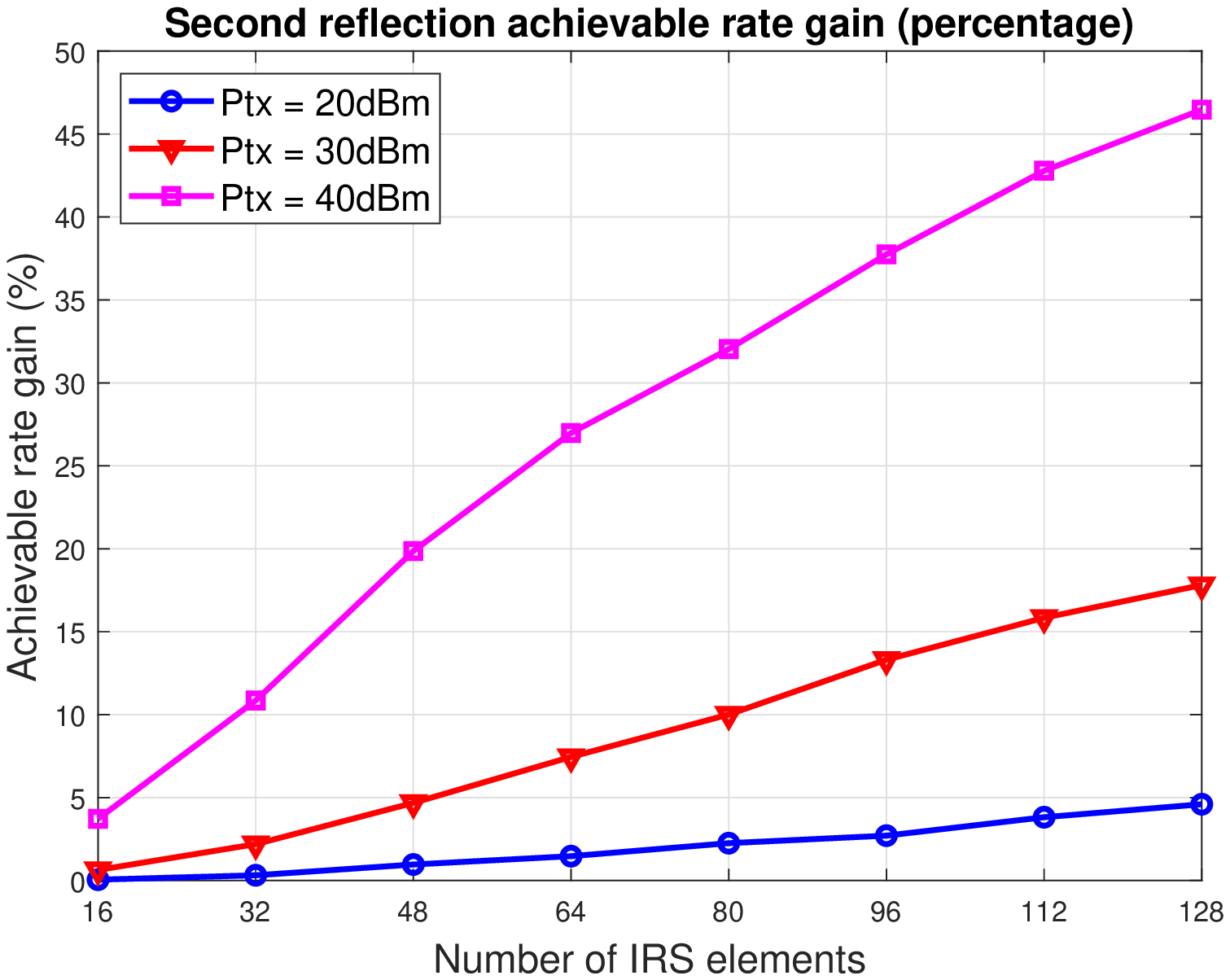}
  \caption{Secondary refl. gain vs number of IRS elements}
  \label{fig:GainvsIRSElements}
\end{figure}

Figure \ref{fig:skipSecondary} compares the minimum achievable rate when both primary and secondary reflections are considered (as in this paper) and when only the primary reflections are considered (as in other aforementioned works \cite{YangVPoor2020}-\cite{zheng2020doubleirs}). We assume the IRSs are equally spaced, i.e., the distance between two adjacent IRSs for 4 and 6 IRSs scenarios are 10 and 6.67 meters, respectively. It can be seen that by properly managing the interference and leveraging the secondary reflections amongst IRSs, the minimum achievable rate is significantly improved. For example, at transmit power 35 dBm, we observe about 1.5 or 2.6 bps/Hz gain when exploiting the secondary reflection for a system with 4 or 6 IRSs, respectively. The improvement becomes more pronounced with higher SNR. This is because at high SNR, the secondary reflections amongst IRSs become more significant while the stronger interference amongst them can be mitigated/managed under our framework. 

The management of inter-IRS interference becomes more critical for dense networks. Specifically, we observe that for a given area, without properly managing the secondary reflections, adding more IRSs may actually degrade the overall system performance. This is due to stronger interference caused by the secondary reflections of the additional IRSs. In the plot, when the secondary reflection is not managed using our proposed method, we observed that the achievable rate of 6 IRSs is worse than that of 4 IRSs at transmit power of 27.0 dBm or higher. With proper managing interference using our method, the performance of 6 IRSs always outperforms that of 4 IRSs as expected.

\begin{figure}[t]
  \centering
  \includegraphics[width=85mm,scale=1.0]{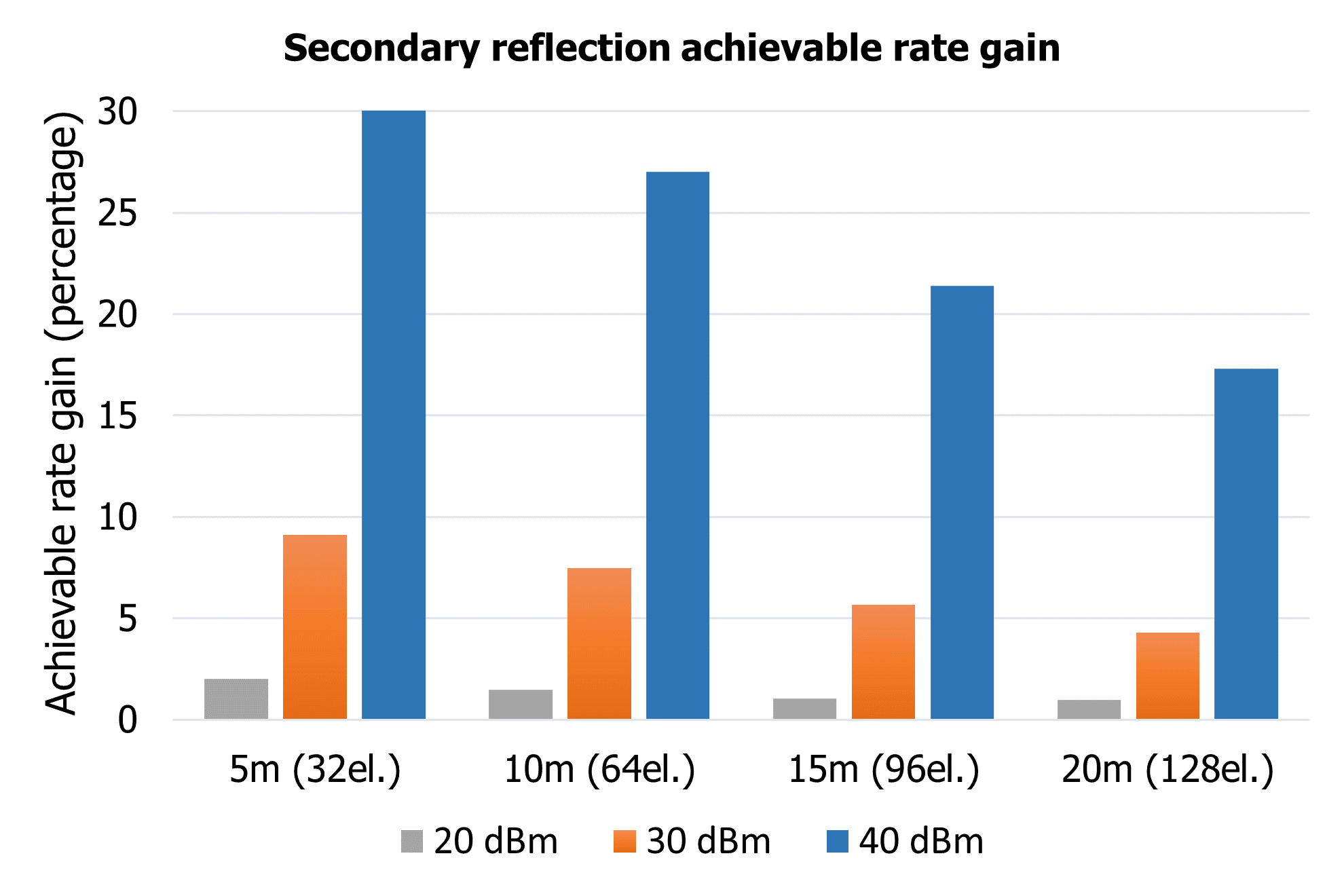}
  \caption{Secondary reflection gain vs IRS distance and number of elements}
  \label{fig:GainvsDistNumElems}
\end{figure}

In Figure \ref{fig:GainvsIRSElements}, we plot the relative gain of the achievable rate with and without considering the secondary reflections for a system with 4 IRSs each with 3 users. As can be seen, the more number of IRS elements, the higher the gain of the secondary reflections. For example, at a transmit power (Ptx) of 30 dBm, we record about 7.5\% and 17.8\% gain for 64 and 128 elements per IRS, respectively. At Ptx=40 dBm, the corresponding gains are 27.0\% and 46.5\%. The gain is increased due to the fact that the more IRS elements, the more reflection energy among IRSs thus increasing in the secondary reflection interference.

In Figure \ref{fig:GainvsDistNumElems}, we plot the achievable rate gain as a result of secondary reflections for a system with 4 IRSs and 3 users while varying the distance between two adjacent IRSs and the number of IRSs elements. When the distance between IRSs are larger, the secondary reflection gains are smaller as expected, but with increased number of IRS elements, the gain can be significantly maintained. For example, at Ptx=40 dBm, the gains for (5m, 32el.) config, i.e., 5m distance and 32 IRS elements and (10m, 64 el.) are 30.2\% and 27.0\%, respectively, and the gains for (15m, 96el.) and (20m, 128el.) are 21.4\% and 17.3\%, respectively.

Lastly, we plotted the achievable rate versus the number of iterations, $I_1$, without early break (i.e., we set $\xi=0$); however, the plot is omitted due to space limitation. We observed that the algorithm reasonably converges after about 50 iterations for Ptx=10 dBm. It requires a slightly higher number of iterations for higher Ptx.

\section{Conclusion}
In this paper aimed to manage the interference and leverage the secondary reflections in systems with multiple IRSs. For that, we considered an uplink multiple IRSs, multiple users MIMO system, where the received beamforming at the BS and the phase-shifts at the IRS elements were co-designed. We derived the closed-form SINR for each user at the BS, and proposed the max-min optimization problem which was shown to be non-convex. An AO algorithm were discussed to solve this problem to obtain a sub-optimal solution. The numerical results showed that by managing the interference and leveraging the secondary reflections amongst IRSs, one can significantly improve the system throughput when more IRSs or more elements per IRS are deployed.

\appendices
\section{Derivation of SINR Formula in Theorem \ref{theorem:GenericSINR}}
\label{App:SINR_proof}
First, we re-write the $\boldsymbol{h}_{l,k}$ as shown in (\ref{eqn:channelmk}) as a function of $\boldsymbol{\theta}_{\bar{l}}$. 

When $l\neq\bar{l}$, we have
\begin{align*}
    \boldsymbol{h}_{l,k} = & \sum_{l'=1}^L\boldsymbol{R}_{l',l,k}\boldsymbol{\theta}_{l'} + \sum_{l'\neq l}^L \sum_{j=1}^{M_{l}} \boldsymbol{Q}_{l',l,k,j} \boldsymbol{\theta}_{l'} \theta_{l,j} \\
    =&\boldsymbol{R}_{\bar{l},l,k}\boldsymbol{\theta}_{\bar{l}} + \sum_{{l'=1},{l'\neq\bar{l}}}^L\boldsymbol{R}_{l',l,k}\boldsymbol{\theta}_{l'} \\
    ~~&+\Bigg(\sum_{j=1}^{M_l} \boldsymbol{Q}_{\bar{l},l,k,j} \theta_{l,j} \Bigg) \boldsymbol{\theta}_{\bar{l}}  + \sum_{j=1}^{M_{l}} \Bigg(\sum_{l'\neq l,l'\neq \bar{l}}^L  \boldsymbol{Q}_{l',l,k,j} \boldsymbol{\theta}_{l'} \Bigg) \theta_{l,j} \\
    =& \left(\boldsymbol{R}_{\bar{l},l,k}+\boldsymbol{T}_{\bar{l},l,k}\right)\boldsymbol{\theta}_{\bar{l}} + \Bigg(\boldsymbol{U}_{\bar{l},l,k} +  \boldsymbol{S}_{\bar{l},l,k} \boldsymbol{\theta}_{l}\Bigg) \numberthis\label{eqn:channelmk_proof1}
\end{align*} 
where we defined
\begin{align}
    \label{eqn:definedS}
    \boldsymbol{S}_{\bar{l},l,k} \triangleq \left[\boldsymbol{s}_{\bar{l},l,k,1},\ldots,\boldsymbol{s}_{\bar{l},l,k,M_l}\right]
\end{align}
with  $\boldsymbol{s}_{\bar{l},l,k,j} \triangleq \sum\limits_{{l'\neq \bar{l}},{l'\neq l}} \boldsymbol{Q}_{l',l,k,j}\boldsymbol{\theta}_{l'}$,
\begin{align}
    \label{eqn:definedTU}
    \boldsymbol{T}_{\bar{l},l,k} \triangleq \sum\limits_{j=1}^{M_{\bar{l}}}\boldsymbol{Q}_{\bar{l},l,k,j}\theta_{\bar{l},j}, ~\mathrm{and }~
    \boldsymbol{U}_{\bar{l},l,k} \triangleq \sum_{{}^{l'=1}_{l'\neq\bar{l}}}^L\boldsymbol{R}_{l',l,k}\boldsymbol{\theta}_{l'}.
\end{align}
When $l=\bar{l}$, we have
\begin{align*}
    \boldsymbol{h}_{l,k} =&\boldsymbol{R}_{\bar{l},l,k} \boldsymbol{\theta}_{\bar{l}} + \sum_{{}^{l'=1}_{l'\neq\bar{l}}}^L \boldsymbol{R}_{l',l,k}\boldsymbol{\theta}_{l'} + \sum_{j=1}^{M_{\bar{l}}} \sum_{l'\neq \bar{l}}^L  \boldsymbol{Q}_{l',l,k,j} \boldsymbol{\theta}_{l'} \theta_{\bar{l},j} \\
    =&\boldsymbol{R}_{\bar{l},l,k}\boldsymbol{\theta}_{\bar{l}} + \boldsymbol{U}_{\bar{l},l,k} +  \boldsymbol{S}_{\bar{l},l,k} \boldsymbol{\theta}_{\bar{l}} \\
    =&\left(\boldsymbol{R}_{\bar{l},l,k} + \boldsymbol{S}_{\bar{l},l,k} \right) \boldsymbol{\theta}_{\bar{l}} + \boldsymbol{U}_{\bar{l},l,k} \numberthis \label{eqn:channelmk_proof_eq}
\end{align*}
Now plug in (\ref{eqn:channelmk_proof1}) and (\ref{eqn:channelmk_proof_eq}) to (\ref{eqn:gamma_mk}), we will obtain (\ref{meqn:rxsinr_mk_reduced}) with $\boldsymbol{q}_{\bar{l},l,k}$ and $\bar{q}_{\bar{l},l,k}$ are given in (\ref{eqn:def_q_lmk}) and (\ref{eqn:barq_lmk}), respectively. \qed




\end{document}